\documentclass[a4paper]{llncs}

\usepackage[english]{babel}
\usepackage[utf8]{inputenc}
\usepackage{hyperref}
\usepackage{enumerate}

\usepackage{amsmath}
\usepackage{amssymb}

\usepackage{listings}
\lstdefinelanguage{pseudo}{
	morekeywords={if, then, else, while, foreach, do, assume, assert, let, return},
	sensitive=false,
	morecomment=[l]{//},
	morecomment=[s]{/*}{*/},
	morestring=[b]",
}
\AtBeginDocument{

}

\newcommand{\loopt}{{\scriptstyle\mathrm{LOOP}}} 
\newcommand{\T}{{\scriptstyle\mathrm{T}}} 
\newcommand{\abovebelow}[2]{(^{#1}_{#2})} 
\newcommand{\tr}[1]{#1^T} 
\newcommand{\prtemplate}{Po\-dels\-ki-Ry\-bal\-chen\-ko ranking template}

\newcommand{\pdftitle}{Ranking Templates for Linear Loops}
\newcommand{\pdfkeywords}{linear lasso, linear loop, termination, linear ranking template, well-founded relation, ordinal number, multiphase ranking function, piecewise ranking function, lexicographic ranking function, Farkas' Lemma, Motzkin's Theorem}
\newcommand{\pdfsubject}{Termination Analysis}
\newcommand{\pdflang}{en-US}

\hypersetup{
  pdftitle={\pdftitle},
  pdfauthor={Jan Leike, Matthias Heizmann},
  pdfkeywords={\pdfkeywords},
  pdfsubject={\pdfsubject},
  pdflang={\pdflang}
}
\title{\pdftitle
\thanks{The final publication is available at \url{http://link.springer.com}.}
}
\author{Jan Leike\inst{1, 2} \and Matthias Heizmann\inst{1}\thanks{This work is supported by the
  German Research Council (DFG) as part of the Transregional Collaborative
  Research Center ``Automatic Verification and Analysis of Complex Systems''
  (SFB/TR14 AVACS)}
}
\institute{
  University of Freiburg, Germany
  \and Max Planck Institute for Software Systems
}
\date{\today}

\begin{document}

\maketitle

\begin{abstract}
We present a new method for the constraint-based synthesis of termination arguments for linear loop programs based on \emph{linear ranking templates}.
Linear ranking templates are parametrized, well-founded relations such that an assignment to the parameters gives rise to a ranking function.
This approach generalizes existing methods and
enables us to use templates for many different ranking functions with affine-linear components.
We discuss templates for multiphase, piecewise, and lexicographic ranking functions.
Because these ranking templates require both strict and non-strict inequalities,
we use Motzkin's Transposition Theorem instead of Farkas Lemma
to transform the generated $\exists\forall$-constraint into an $\exists$-constraint.
\end{abstract}


\section{Introduction}
\label{sec-introduction}

The scope of this work is the constraint-based synthesis of termination arguments.
In our setting, we consider \emph{linear loop programs}, which are specified by a boolean combination of affine-linear inequalities over the program variables.
This allows for both, deterministic and non-deterministic updates of the program variables.
An example of a linear loop program is given in \autoref{fig-multiphase-introduction}.

\begin{figure}[t]
\begin{center}
\begin{minipage}{35mm}
\begin{lstlisting}
while ($q > 0$):
    $q$ := $q - y$;
    $y$ := $y + 1$;
\end{lstlisting}
\end{minipage}
\begin{minipage}{35mm}
\begin{align*}
q &> 0  \\
\land\; q' &= q - y \\
\land\; y' &= y + 1
\end{align*}
\end{minipage}
\vspace{-3mm}
\end{center}
\caption{
A linear loop program given as program code (left) and
as a formula defining a binary relation (right).
}\label{fig-multiphase-introduction}
\end{figure}

Usually, linear loop programs do not occur as stand-alone programs.
Instead, they are generated as a finite representation of an infinite path in a control flow graph.
For example, in (potentially spurious) counterexamples in termination analysis~\cite{CPR06,HLNR10,KSTTW08,KSTW10,PodelskiR04,PodelskiR05},
non-termination analysis~\cite{GHMRX08},
stability analysis~\cite{CFKP11,PW07},
or cost analysis~\cite{AAGP11,GZ10}.

We introduce the notion of \emph{linear ranking templates}
(\autoref{sec-templates}).
These are parameterized relations specified by linear inequalities such that
any assignment to the parameters yields a well-founded relation.
This notion is general enough to encompass all existing methods for linear loop programs
that use constraint-based synthesis of ranking functions of various kinds
(see \autoref{sec-related-work} for an assessment).
Moreover, ours is the first method for synthesis of lexicographic ranking functions
that does not require a mapping between loop disjuncts and lexicographic components.

In this paper we present the following linear ranking templates.
\begin{itemize}
\item The \emph{multiphase ranking template} specifies a ranking function
that proceeds through a fixed number of phases in the program execution.
Each phase is ranked by an affine-linear function; when this function
becomes non-positive, we move on to the next phase (\autoref{ssec-rt-multiphase}).
\item The \emph{piecewise ranking template} specifies a ranking function
that is a piecewise affine-linear function
with affine-linear predicates to discriminate between the pieces (\autoref{ssec-rt-pw}).
\item The \emph{lexicographic ranking template} specifies a lexicographic ranking function
that corresponds to a tuple of affine-linear functions together with a lexicographic ordering on the tuple (\autoref{ssec-rt-lex}).
\end{itemize}
These linear ranking templates can be used as a 'construction kit'
for composing linear ranking templates that enable more complex ranking functions (\autoref{ssec-composition}).
Moreover, variations on the linear ranking templates presented here can be used and completely different templates could be conceived.

Our method is described in \autoref{sec-rf-synthesis} and can be summarized as follows.
The input is a linear loop program as well as a linear ranking template.
From these we construct a constraint to the parameters of the template.
With Motzkin's Theorem we can transform the constraint into a purely existentially quantified constraint (\autoref{ssec-constraint}).
This $\exists$-constraint is then passed to an SMT solver which checks it for satisfiability.
A positive result implies that the program terminates.
Furthermore, a satisfying assignment will yield a ranking function,
which constitutes a termination argument for the given linear loop program.

Related approaches invoke Farkas' Lemma for the transformation into $\exists$-constraints~\cite{ADFG10,BMS05linrank,BMS05polyrank,CSS03,HHLP13,PR04,Rybalchenko10,SSM04}.
The piecewise and the lexicographic ranking template contain
both strict and non-strict inequalities, yet only non-strict inequalities can be transformed using Farkas' Lemma.
We solve this problem by introducing the use of
\hyperref[thm-motzkin]{Motzkin's Transposition Theorem}, a generalization of Farkas' Lemma.
As a side effect, this also enables both, strict and non-strict inequalities in the program syntax.
To our knowledge, all of the aforementioned methods can be improved by the application of Motzkin's Theorem instead of Farkas' Lemma.

Our method is complete in the following sense.
If there is a ranking function of the form specified by the given linear ranking template,
then our method will discover this ranking function.
In other words, the existence of a solution is never lost in the process of transforming the constraint.

In contrast to some related methods~\cite{HHLP13,PR04}
the constraint we generate is not linear,
but rather a nonlinear algebraic constraint.
Theoretically, this constraint can be decided in exponential time~\cite{GV88}.
Much progress on nonlinear SMT solvers has been made and
present-day implementations routinely solve nonlinear constraints of various sizes~\cite{JM12}.

A related setting to linear loop programs are \emph{linear lasso programs}. These consist of a linear loop program and a program stem,
both of which are specified by boolean combinations of affine-linear inequalities over the program variables.
Our method can be extended to linear lasso programs through the addition of affine-linear inductive invariants,
analogously to related approaches~\cite{BMS05linrank,CSS03,HHLP13,SSM04}.


\section{Preliminaries}
\label{sec-preliminaries}

In this paper we use $\mathbb{K}$ to denote a field that is either the rational numbers $\mathbb{Q}$ or
the real numbers $\mathbb{R}$.
We use ordinal numbers according to the definition in \cite{Jech06}.
The first infinite ordinal is denoted by $\omega$; the finite ordinals coincide with the natural numbers, therefore we will use them interchangeably.

\subsection{Motzkin's Transposition Theorem}
\label{ssec-motzkin}

Intuitively, Motzkin's Transposition Theorem~\cite[Corollary 7.1k]{Schrijver99}
states that a given system of linear inequalities has no solution if
and only if a contradiction can be derived via a positive linear
combination of the inequalities.

\begin{theorem}[Motzkin's Transposition Theorem]
\label{thm-motzkin}
For $A \in \mathbb{K}^{m \times n}$, $C \in \mathbb{K}^{\ell \times n}$,
$b \in \mathbb{K}^m$, and $d \in \mathbb{K}^\ell$,
the formulae \eqref{eq-motzkin1} and \eqref{eq-motzkin2} are equivalent.
\begin{align}
&\hspace{14.7mm} \forall x \in \mathbb{K}^n.\;
  \neg (Ax \leq b \;\land\; Cx < d)
\tag{M1}\label{eq-motzkin1} \\[1.5mm]
&\begin{aligned}
\exists \lambda \in \mathbb{K}^m \;\exists \mu \in \mathbb{K}^\ell.\;
  &\lambda \geq 0 \;\land\; \mu \geq 0 \\
  \land\; &\tr{\lambda} A + \tr{\mu} C = 0
  \;\land\;  \tr{\lambda} b + \tr{\mu} d \leq 0 \\
  \land\; &( \tr{\lambda} b < 0 \;\lor\; \mu \neq 0 )
\end{aligned}\tag{M2}\label{eq-motzkin2}
\end{align}
\end{theorem}

If $\ell$ is set to $1$ in \autoref{thm-motzkin}, we obtain the affine version of Farkas' Lemma~\cite[Corollary 7.1h]{Schrijver99}.
Therefore Motzkin's Theorem is strictly superior to Farkas' Lemma, as it allows for a combination of both strict and non-strict inequalities.
Moreover, it is logically optimal in the sense that it enables transformation of \emph{any} universally quantified formula from the theory of linear arithmetic.

\subsection{Linear Loop Programs}
\label{ssec-loops}

In this work, we consider programs that consist of a single loop.
We use binary relations over the program's states to define its transition relation.

We denote by $x$ the vector of $n$ variables $(x_1, \ldots, x_n)^T \in \mathbb{K}^n$ corresponding to program states and
by $x' = (x_1', \ldots, x_n')^T \in \mathbb{K}^n$ the variables of the next state.

\begin{definition}[Linear loop program]\label{def-linear-loops}
A \emph{linear loop program} $\loopt(x, x')$ is a binary relation defined by a formula with the free variables $x$ and $x'$ of the form
\[
\bigvee_{i \in I} \big( A_i\abovebelow{x}{x'} \leq b_i \;\land\; C_i\abovebelow{x}{x'} < d_i \big)
\]
for some finite index set $I$, some matrices $A_i \in \mathbb{K}^{n \times m_i}$,
$C_i \in \mathbb{K}^{n \times k_i}$,
and some vectors $b_i \in \mathbb{K}^{m_i}$ and $d_i \in \mathbb{K}^{k_i}$.
The linear loop program $\loopt(x, x')$ is called \emph{conjunctive}, iff there is only one disjunct, i.e., $\#I = 1$.
\end{definition}
Geometrically the relation $\loopt$ corresponds to a union of convex polyhedra.

\begin{definition}[Termination]\label{def-termination}
We say that a linear loop program $\loopt(x, x')$ \emph{terminates} iff
the relation $\loopt(x, x')$ is well-founded.
\end{definition}

\begin{example}\label{ex-running1}
Consider the following program code.
\begin{center}
\begin{minipage}{45mm}
\begin{lstlisting}
while ($q > 0$):
    if ($y > 0$):
        $q$ := $q - y - 1$;
    else:
        $q$ := $q + y - 1$;
\end{lstlisting}
\end{minipage}
\end{center}
We represent this as the following linear loop program:
\begin{align*}
&(q > 0 \;\land\; y > 0
  \;\land\; y' = y \;\land\; q' = q - y - 1) \\
\lor\; &(q > 0 \;\land\; y \leq 0
  \;\land\; y' = y \;\land\; q' = q + y - 1)
\end{align*}
This linear loop program is not conjunctive.
Furthermore, there is no infinite sequence of states $x_0, x_1, \ldots$ such that for all $i \geq 0$, the two successive states $(x_i, x_{i+1})$ are contained in the relation $\loopt$.
Hence the relation $\loopt(x, x')$ is well-founded and the linear loop program terminates.
\end{example}


\section{Ranking Templates}
\label{sec-templates}

A \emph{ranking template} is a template for a well-founded relation.
More specifically, it is a parametrized formula defining a relation that is well-founded for all assignments to the parameters.
If we show that a given program's transition relation $\loopt$ is a subset of an instance of this well-founded relation, it must be well-founded itself and thus we have a proof for the program's termination.
Moreover, an assignment to the parameters of the template gives rise to a ranking function.
In this work, we consider ranking templates that can be encoded in linear arithmetic.

We call a formula whose free variables contain $x$ and $x'$ a \emph{relation template}.
Each free variable other than $x$ and $x'$ in a relation template is called \emph{parameter}.
Given an assignment $\nu$ to all parameter variables of a relation template $\T(x, x')$,
the evaluation $\nu(\T)$ is called an \emph{instantiation of the relation template $\T$}.
We note that each instantiation of a relation template $\T(x, x')$ defines a binary relation.

When specifying templates, we use parameter variables to define affine-linear functions.
For notational convenience, we will write $f(x)$ instead of the term $\tr{s_f} x + t_f$,
where $s_f \in \mathbb{K}^n$ and $t_f \in \mathbb{K}$ are parameter variables.
We call $f$ an \emph{affine-linear function symbol}.

\begin{definition}[Linear ranking template]
\label{def-linear-rt}
Let $\T(x, x')$ be a template with parameters $D$ and affine-linear function symbols $F$ that can be written as a boolean combination of atoms of the form
\begin{align*}
\sum_{f \in F} \big( \alpha_f \cdot f(x) + \beta_f \cdot f(x') \big)
+ \sum_{d \in D} \gamma_d \cdot d \;\rhd\; 0,
\end{align*}
where $\alpha_f, \beta_f, \gamma_d \in \mathbb{K}$ are constants and
$\rhd \in \{ \geq, > \}$.
We call $\T$ a \emph{linear ranking template} over $D$ and $F$ iff
every instantiation of $\T$ defines a well-founded relation.
\end{definition}

\begin{example}\label{ex-rt-pr}
We call the following template with parameters $D = \{ \delta \}$ and affine-linear function symbols $F = \{ f \}$ the
\emph{\prtemplate}~\cite{PR04}.
\begin{align}
\delta > 0 \;\land\; f(x) > 0 \;\land\; f(x') < f(x) - \delta
\label{eq-rt-pr}
\end{align}
In the remainder of this section, we introduce a formalism that allows us to show that every instantiation of the {\prtemplate} defines a well-founded relation.
Let us now check the additional syntactic requirements for \ref{eq-rt-pr} to be a linear ranking template:
\begin{align*}
\delta > 0 \;&\equiv\;
  \big( 0 \cdot f(x) + 0 \cdot f(x') \big)
  + 1 \cdot \delta > 0 \\
f(x) > 0 \;&\equiv\;
  \big( 1 \cdot f(x) + 0 \cdot f(x') \big)
  + 0 \cdot \delta > 0 \\
f(x') < f(x) - \delta \;&\equiv\;
  \big( 1 \cdot f(x) + (- 1) \cdot f(x') \big)
  + (-1) \cdot \delta > 0
\end{align*}
\end{example}

The next lemma states that we can prove termination of a given linear loop program
by checking that this program's transition relation is included in an instantiation of a linear ranking template.

\begin{lemma}\label{lem-termination}
Let $\loopt$ be a linear loop program and
let $\T$ be a linear ranking template with parameters $D$ and affine-linear function symbols $F$.
If there is an assignment $\nu$ to $D$ and $F$ such that the formula
\begin{align}
\forall x, x'.\;
\big( \loopt(x, x') \rightarrow \nu(\T)(x, x') \big)
\label{eq-rt}
\end{align}
is valid, then the program $\loopt$ terminates.
\end{lemma}
\begin{proof}
By definition, $\nu(\T)$ is a well-founded relation and
\eqref{eq-rt} is valid iff
the relation $\loopt$ is a subset of $\nu(\T)$.
Thus $\loopt$ must be well-founded.
\qed
\end{proof}

In order to establish that a formula conforming to the syntactic
requirements is indeed a ranking template, we must show the well-foundedness of its instantiations.
According to the following lemma, we can do this by showing that an assignment to $D$ and $F$ gives rise to a ranking function.
A similar argument was given in \cite{Ben-Amram09};
we provide significantly shortened proof by use of the Recursion Theorem, along the lines of \cite[Example 6.12]{Jech06}.

\begin{definition}[Ranking Function]\label{def-ranking}
Given a binary relation $R$ over a set $\Sigma$, we call a function $\rho$ from $\Sigma$ to some ordinal $\alpha$ a \emph{ranking function} for $R$ iff for all $x,x'\in \Sigma$ the following implication holds.
\[
(x, x') \in R \;\Longrightarrow\; \rho(x) > \rho(x')
\]
\end{definition}

\begin{lemma}\label{lem-ranking}
A binary relation $R$ is well-founded if and only if there exists a ranking function for $R$.
\end{lemma}
\begin{proof}
Let $\rho$ be a ranking function for $R$.
The image of a sequence decreasing with respect to $R$ under $\rho$ is a
strictly decreasing ordinal sequence.
Because the ordinals are well-ordered, this sequence cannot be infinite.

Conversely, the graph $G = (\Sigma, R)$ with vertices $\Sigma$ and edges $R$ is acyclic by assumption.
Hence the function $\rho$ that assigns to every element of $\Sigma$ an ordinal number such that
$\rho(x) = \sup\, \{ \rho(x') + 1 \mid (x, x') \in R \}$
is well-defined and exists due to the Recursion Theorem~\cite[Theorem 6.11]{Jech06}.
\qed
\end{proof}

\begin{example}\label{ex-running2}
Consider the terminating linear loop program $\loopt$ from \autoref{ex-running1}.
A ranking function for $\loopt$ is $\rho: \mathbb{R}^2 \to \omega$, defined as follows.
\begin{align*}
\rho(q, y) =
\begin{cases}
\lceil q \rceil, &\text{if } q > 0,
\text{ and} \\
0 & \text{otherwise.}
\end{cases}
\end{align*}
Here $\lceil \cdot \rceil$ denotes the ceiling function that assigns to
every real number $r$ the smallest natural number that is larger or equal to $r$.
Since we consider the natural numbers to be a subset of the ordinals, the ranking function $\rho$ is well-defined.
\end{example}

We use assignments to a template's {parameters and affine-linear function symbols} to construct a ranking function.
These functions are real-valued and we will transform them into functions with codomain $\omega$ in the following way.

\begin{definition}\label{def-ordinal-ranking}
Given an affine-linear function $f$ and a real number $\delta > 0$ called
the \emph{step size}, we define the \emph{ordinal ranking
equivalent} of $f$ as
\begin{align*}
\widehat{f}(x) =
\begin{cases}
\left\lceil \frac{f(x)}{\delta} \right\rceil, &\text{if } f(x) > 0,
\text{ and} \\
0 & \text{otherwise.}
\end{cases}
\end{align*}
\end{definition}

For better readability we used this notation which does not explicitly refer to $\delta$.
In our presentation the step size $\delta$ is always clear from the context in which an ordinal ranking equivalent $\widehat{f}$ is used. 

\begin{example}\label{ex-running3}
Consider the linear loop program $\loopt(x, x')$ from \autoref{ex-running1}.
For $\delta = \frac{1}{2}$ and $f(q) = q + 1$,
the ordinal ranking equivalent of $f$ with step size $\delta$ is
\begin{align*}
\widehat{f}(q, y) =
\begin{cases}
\lceil 2(q + 1) \rceil, & \text{if } q + 1 > 0, \text{ and} \\
0 & \text{otherwise.}
\end{cases}
\end{align*}
\end{example}

The assignment from \autoref{ex-running3} to $\delta$ and $f$ makes the implication \eqref{eq-rt} valid.
In order to invoke \autoref{lem-termination} to show that the linear loop program given in \autoref{ex-running1} terminates,
we need to prove that the {\prtemplate} is a linear ranking template.
We use the following technical lemma.

\begin{lemma}\label{lem-ordinal-ranking}
Let $f$ be an affine-linear function of step size $\delta > 0$ and
let $x$ and $x'$ be two states.
If $f(x) > 0$ and $f(x) - f(x') > \delta$,
then $\widehat{f}(x) > 0$ and $\widehat{f}(x) > \widehat{f}(x')$.
\end{lemma}
\begin{proof}
From $f(x) > 0$ follows that $\widehat{f}(x) > 0$.
Therefore $\widehat{f}(x) > \widehat{f}(x')$ in the case $\widehat{f}(x') = 0$.
For $\widehat{f}(x') > 0$, we use the fact that $f(x) - f(x') > \delta$
to conclude that $\frac{f(x)}{\delta} - \frac{f(x')}{\delta} > 1$ and
hence $\widehat{f}(x') > \widehat{f}(x)$.
\qed
\end{proof}

An immediate consequence of this lemma is that
the {\prtemplate} is a linear ranking template:
any assignment $\nu$ to $\delta$ and $f$ satisfies the requirements of \autoref{lem-ordinal-ranking}.
Consequently, $\widehat{f}$ is a ranking function for $\nu(\T)$,
and by \autoref{lem-ranking} this implies that $\nu(\T)$ is well-founded.


\section{Examples of Ranking Templates}
\label{sec-ranking-templates-examples}

\subsection{Multiphase Ranking Template}
\label{ssec-rt-multiphase}

The multiphase ranking template is targeted at
programs that go through a finite number of phases in their execution.
Each phase is ranked with an affine-linear function and
the phase is considered to be completed once this function becomes non-positive.

\begin{example}\label{ex-2phase-1}
Consider the linear loop program from \autoref{fig-multiphase-introduction}.
Every execution can be partitioned into two phases:
first $y$ increases until it is positive and
then $q$ decreases until the loop condition $q > 0$ is violated.
Depending on the initial values of $y$ and $q$,
either phase might be skipped altogether.
\end{example}

\begin{definition}[Multiphase Ranking Template]\label{def-rt-multiphase}
We define the \emph{$k$-phase ranking template} with
parameters $D = \{ \delta_1, \ldots, \delta_k \}$ and
affine-linear function symbols $F = \{ f_1, \ldots, f_k \}$ as follows.
\begin{align}
&\bigwedge_{i=1}^k \delta_i > 0 \label{eq-rt-multiphase1} \\
\land\; &\bigvee_{i=1}^k f_i(x) > 0 \label{eq-rt-multiphase2} \\
\land\; &f_1(x') < f_1(x) - \delta_1 \label{eq-rt-multiphase3} \\
\land\; &\bigwedge_{i=2}^k \Big( f_i(x') < f_i(x) - \delta_i
  \;\lor\; f_{i-1}(x) > 0 \Big) \label{eq-rt-multiphase4}
\end{align}
\end{definition}

We say that the multiphase ranking function given by
an assignment to $f_1, \ldots, f_k$ and $\delta_1, \ldots, \delta_k$
\emph{is in phase $i$},
iff $f_i(x) > 0$ and $f_j(x) \leq 0$ for all $j < i$.
The condition \eqref{eq-rt-multiphase2} states that there is always some $i$ such that
the multiphase ranking function is in phase $i$.
\eqref{eq-rt-multiphase3} and \eqref{eq-rt-multiphase4} state that if we are in a phase $\geq i$,
then $f_i$ has to be decreasing by at least $\delta_i > 0$.
Note that the $1$-phase ranking template coincides with the {\prtemplate}.

Multiphase ranking functions are related to \emph{eventually negative expressions} introduced by Bradley, Manna, and Sipma~\cite{BMS05polyrank}.
However, in contrast to our approach,
they require a template tree that specifies in detail
how each loop transition interacts with each phase.

\begin{lemma}\label{lem-rt-multiphase}
The $k$-phase ranking template is a linear ranking template.
\end{lemma}
\begin{proof}
The $k$-phase ranking template conforms to the syntactic requirements
to be a linear ranking template.
Let an assignment to the parameters $D$ and the affine-linear function symbols $F$ of the $k$-phase template be given.
Consider the following ranking function with codomain $\omega \cdot k$.
\begin{align}
\rho(x) =
\begin{cases}
\omega \cdot (k - i) + \widehat{f_i}(x) & \text{if }
f_j(x) \leq 0 \text{ for all } j < i \text{ and } f_i(x) > 0, \\
0 & \text{otherwise.}
\end{cases}
\label{eq-rt-multiphase-rf}
\end{align}
Let $(x, x') \in \T$.
By \autoref{lem-ranking}, we need to show that $\rho(x') < \rho(x)$.
From \eqref{eq-rt-multiphase2} follows that $\rho(x) > 0$.
Moreover, there is an $i$ such that $f_i(x) > 0$ and $f_j(x) \leq 0$ for all $j < i$.
By \eqref{eq-rt-multiphase3} and \eqref{eq-rt-multiphase4},
$f_j(x') \leq 0$ for all $j < i$ because
$f_j(x') < f_j(x) - \delta_j \leq 0 - \delta_j \leq 0$, since
$f_\ell(x) \leq 0$ for all $\ell < j$.

If $f_i(x') \leq 0$, then $\rho(x') \leq \omega \cdot (k - i)
< \omega \cdot (k - i) + \widehat{f_i}(x) = \rho(x)$.
Otherwise, $f_i(x') > 0$ and from \eqref{eq-rt-multiphase4} follows $f_i(x') < f_i(x) - \delta_i$.
By \autoref{lem-ordinal-ranking},
$\widehat{f_i}(x) > \widehat{f_i}(x')$ for the ordinal ranking equivalent of $f_i$ with step size $\delta_i$.
Hence
\[
\rho(x') = \omega \cdot (k - i) + \widehat{f_i}(x')
  < \omega \cdot (k - i) + \widehat{f_i}(x)
  = \rho(x). \eqno\qed
\]
\end{proof}

\begin{example}\label{ex-2phase-2}
Consider the program from \autoref{fig-multiphase-introduction}.
The assignment
\begin{center}
$f_1(q, y) = 1 - y, \qquad
f_2(q, y) = q + 1,
\qquad \delta_1 = \delta_2 = \frac{1}{2}$
\end{center}
yields a $2$-phase ranking function for this program.
\end{example}

\begin{example}\label{ex-multiphase-rotation}
There are terminating conjunctive linear loop programs that do not have a
multiphase ranking function:
\begin{align*}
q > 0 \;\land\; q' = q + z - 1 \;\land\; z' = -z
\end{align*}
Here, the sign of $z$ is alternated in each iteration. The function $\rho(q, y, z) = \lceil q \rceil$ is decreasing in every second iteration, but not decreasing in each iteration.
\end{example}

\begin{example}\label{ex-multiphase-complexity}
Consider the following linear loop program.
\begin{align*}
  &(q > 0 \;\land\; y > 0 \;\land\; y' = 0) \\
\lor\; &(q > 0 \;\land\; y \leq 0 \;\land\; y' = y - 1 \;\land\; q' = q - 1)
\end{align*}
For a given input, we cannot give an upper bound on the execution time:
starting with $y > 0$,
after the first loop execution, $y$ is set to $0$ and
$q$ is set to \emph{some arbitrary value}, as no restriction to $q'$ applies in the first disjunct.
In particular, this value does not depend on the input.
The remainder of the loop execution then takes $\lceil q \rceil$ iterations to terminate.

However we can prove the program's termination with
the $2$-phase ranking function constructed from
$f_1(q, y) = y$ and $f_2(q, y) = q$.
\end{example}

\subsection{Piecewise Ranking Template}
\label{ssec-rt-pw}

The piecewise ranking template formalizes a ranking function that is defined piecewise using affine-linear
predicates to discriminate the pieces.

\begin{definition}[Piecewise Ranking Template]\label{def-rt-pw}
We define the \emph{$k$-piece ranking template} with
parameters $D = \{ \delta \}$ and
affine-linear function symbols $F = \{ f_1, \ldots, f_k, g_1, \ldots, g_k \}$
as follows.
\begin{align}
&\delta > 0 \label{eq-rt-pw1} \\
\land\; &\bigwedge_{i=1}^k \bigwedge_{j=1}^k \Big( g_i(x) < 0
  \;\lor\; g_j(x') < 0 \;\lor\; f_j(x') < f_i(x) - \delta \Big) \label{eq-rt-pw2} \\
\land\; &\bigwedge_{i=1}^k f_i(x) > 0 \label{eq-rt-pw3} \\
\land\; &\bigvee_{i=1}^k g_i(x) \geq 0 \label{eq-rt-pw4}
\end{align}
\end{definition}
We call the affine-linear function symbols $\{ g_i \mid 1 \leq i \leq k \}$ \emph{discriminators} and
the affine-linear function symbols $\{ f_i \mid 1 \leq i \leq k \}$ \emph{ranking pieces}.

The disjunction \eqref{eq-rt-pw4} states that the discriminators cover all states;
in other words, the piecewise defined ranking function is a total function.
Given the $k$ different pieces $f_1, \ldots, f_k$ and a state $x$,
we use $f_i$ as a ranking function only if $g_i(x) \geq 0$ holds.
This choice need not be unambiguous;
the discriminators may overlap.
If they do, we can use any one of their ranking pieces.
According to \eqref{eq-rt-pw3}, all ranking pieces are positive-valued and
by \eqref{eq-rt-pw2} piece transitions are well-defined:
the rank of the new state is always less than the rank any of the ranking pieces assigned to the old state.

\begin{lemma}\label{lem-rt-pw}
The $k$-piece ranking template is a linear ranking template.
\end{lemma}
\begin{proof}
The $k$-piece ranking template conforms to the syntactic requirements to be a linear ranking template.
Let an assignment to the parameter $\delta$ and the affine-linear function symbols $F$ of the $k$-piece template be given.
Consider the following ranking function with codomain $\omega$.
\begin{align}
\rho(x) = \max \big\{ \widehat{f_i}(x) \mid g_i(x) \geq 0 \big\}
\end{align}
The function $\rho$ is well-defined, because according to \eqref{eq-rt-pw4},
the set $\{ \widehat{f_i}(x) \mid g_i(x) \geq 0 \}$ is not empty.
Let $(x, x') \in \T$ and let $i$ and $j$ be indices such that
$\rho(x) = \widehat{f_i}(x)$ and $\rho(x') = \widehat{f_j}(x')$.
By definition of $\rho$, we have that $g_i(x) \geq 0$ and $g_j(x) \geq 0$, and
\eqref{eq-rt-pw2} thus implies $f_j(x') < f_i(x) - \delta$.
According to \autoref{lem-ordinal-ranking} and \eqref{eq-rt-pw3},
this entails $\widehat{f_j}(x') < \widehat{f_i}(x)$ and therefore $\rho(x') < \rho(x)$.
\autoref{lem-ranking} now implies that $\T$ is well-founded.
\qed
\end{proof}

\begin{example}\label{ex-rt-pw}
Consider the following linear loop program.
\begin{align*}
&(q > 0 \;\land\; p > 0 \;\land\; q < p \;\land\; q' = q - 1) \\
\lor\; &(q > 0 \;\land\; p > 0 \;\land\; p < q \;\land\; p' = p - 1)
\end{align*}
In every loop iteration, the minimum of $p$ and $q$ is decreased by $1$ until it becomes negative.
Thus, this program is ranked by the 2-piece ranking function constructed from $f_1(p, q) = p$ and $f_2(p, q) = q$
with step size $\delta = \frac{1}{2}$ and
discriminators $g_1(p, q) = q - p$ and $g_2(p, q) = p - q$.
Moreover, this program does not have a multiphase or lexicographic ranking function:
both $p$ and $q$ may increase without bound during program execution due to non-determinism and
the number of switches between $p$ and $q$ being the minimum value is also unbounded.
\end{example}

\subsection{Lexicographic Ranking Template}
\label{ssec-rt-lex}

Lexicographic ranking functions consist of lexicographically ordered components of affine-linear functions.
A state is mapped to a tuple of values such that the loop transition leads to a decrease with respect to the lexicographic ordering for this tuple.
Therefore no function may increase unless a function of a lower index decreases.
Additionally, at every step, there must be at least one function that decreases.

Several different definitions for lexicographic ranking functions have been utilized~\cite{ADFG10,BAG13,BMS05linrank};
a comparison can be found in \cite{BAG13}.
Each of these definitions for lexicographic linear ranking functions can be formalized using linear ranking templates;
in this publication we are following the definition of \cite{ADFG10}.

\begin{definition}[Lexicographic Ranking Template]\label{def-rt-lex}
We define the \emph{$k$-lexi\-co\-gra\-phic ranking template}
with parameters $D = \{ \delta_1, \ldots, \delta_k \}$ and
affine-linear function symbols $F = \{ f_1, \ldots, f_k \}$ as follows.
\begin{align}
&\bigwedge_{i=1}^{k} \delta_i > 0 \label{eq-rt-lex1} \\
\land\; &\bigwedge_{i=1}^k f_i(x) > 0 \label{eq-rt-lex2} \\
\land\; &\bigwedge_{i=1}^{k-1} \Big( f_i(x') \leq f_i(x)
  \;\lor\; \bigvee_{j=1}^{i-1} f_j(x') < f_j(x) - \delta_j \Big) \label{eq-rt-lex3} \\
\land\; &\bigvee_{i=1}^k f_i(x') < f_i(x) - \delta_i \label{eq-rt-lex4}
\end{align}
\end{definition}
The conjunction \eqref{eq-rt-lex2} establishes that all lexicographic components $f_1, \ldots, f_k$ have positive values.
In every step, at least one component must decrease according to \eqref{eq-rt-lex4}.
From \eqref{eq-rt-lex3} follows that
all functions corresponding to components of smaller index than the decreasing function may increase.

\begin{lemma}\label{lem-rt-lex}
The $k$-lexicographic ranking template is a linear ranking template.
\end{lemma}
\begin{proof}
The $k$-lexicographic ranking template conforms to the syntactic requirements to be a linear ranking template.
Let an assignment to the parameters $D$ and the affine-linear function symbols $F$ of the $k$-lexicographic template be given.
Consider the following ranking function with codomain $\omega^k$.
\begin{align}
\rho(x) = \sum_{i=1}^k \omega^{k-i} \cdot \widehat{f_i}(x)
\end{align}
Let $(x, x') \in \T$.
From \eqref{eq-rt-lex2} follows $f_j(x) > 0$ for all $j$, so $\rho(x) > 0$.
By \eqref{eq-rt-lex4} and \autoref{lem-ordinal-ranking},
there is a minimal $i$ such that $\widehat{f_i}(x') < \widehat{f_i}(x)$.
According to \eqref{eq-rt-lex3},
$\widehat{f_1}(x') \leq \widehat{f_1}(x)$ and hence
inductively $\widehat{f_j}(x') \leq \widehat{f_j}(x)$ for all $j < i$,
since $i$ was minimal.
\begin{align*}
\rho(x') &= \sum_{j=1}^k \omega^{k-j} \cdot \widehat{f_j}(x')
  \leq \sum_{j=1}^{i-1} \omega^{k-j} \cdot \widehat{f_j}(x)
    + \sum_{j=i}^k \omega^{k-j} \cdot \widehat{f_j}(x') \\
  &< \sum_{j=1}^{i-1} \omega^{k-j} \cdot \widehat{f_j}(x)
    + \omega^{k-i} \cdot \widehat{f_i}(x)
  \leq \rho(x)
\end{align*}
Therefore \autoref{lem-ranking} implies that $\T$ is well-founded. \qed
\end{proof}


\subsection{Composition of Templates}
\label{ssec-composition}

The multiphase ranking template, the piecewise ranking template, and the lexicographic ranking template defined in the previous sections
can be used as a 'construction kit' for more general linear ranking templates.
Each of our templates contains lower bounds (\eqref{eq-rt-multiphase2}, \eqref{eq-rt-pw3}, \eqref{eq-rt-pw4}, and \eqref{eq-rt-lex2})
and decreasing behavior (\eqref{eq-rt-multiphase3}, \eqref{eq-rt-multiphase4}, \eqref{eq-rt-pw2}, \eqref{eq-rt-lex3}, and \eqref{eq-rt-lex4}).
We can compose templates by replacing the lower bound conditions and decreasing behavior conditions to affine-linear function symbols in our linear ranking templates with the corresponding conditions of another template.
This is possible because linear ranking templates allow arbitrary boolean combination of inequalities and
are closed under this kind of substitution.
For example, we can construct a template for a lexicographic ranking function
whose lexicographic components are multiphase functions instead of affine-linear functions
(see \autoref{fig-rt-multilex}).
This encompasses the approach applied by Bradley et al.~\cite{BMS05polyrank}.

\begin{figure}[t]
\begin{align*}
&\bigwedge_{i=1}^{k} \bigwedge_{j=1}^\ell \delta_{i,j} > 0 \\
\land\; &\bigwedge_{i=1}^k \bigvee_{j=1}^\ell f_{i,j}(x) > 0 \\
\land\; &\bigwedge_{i=1}^{k-1} \Big(
  \Big( f_{i,1}(x') \leq f_{i,1}(x)
    \;\land\; \bigwedge_{j=2}^\ell \big( f_{i,j}(x') \leq f_{i,j}(x)
      \;\lor\; f_{i,j-1}(x) > 0 \big)
  \Big) \\
  &\;\lor\; \bigvee_{t=1}^{i-1}
  \Big( f_{t,1}(x') < f_{t,1}(x) - \delta_{t,1}
    \;\land\; \bigwedge_{j=2}^\ell \big( f_{t,j}(x') < f_{t,j}(x) - \delta_{t,j}
      \;\lor\; f_{t,j-1}(x) > 0 \big)
  \Big)
\Big) \\
\land\; &\bigvee_{i=1}^k \Big( f_{i,1}(x') < f_{i,1}(x) - \delta_{i,1}
  \;\land\; \bigwedge_{j=2}^\ell \big( f_{i,j}(x') < f_{i,j}(x) - \delta_{i,j}
    \;\lor\; f_{i,j-1}(x) > 0 \big)
\Big)
\end{align*}
\caption{
A $k$-lexicographic ranking template with $\ell$ phases in each lexicographic component
with the parameters$D = \{ \delta_{i,j} \}$ and
affine-linear function symbols $F = \{ f_{i,j} \}$.
}\label{fig-rt-multilex}
\end{figure}


\section{Synthesizing Ranking Functions}
\label{sec-rf-synthesis}

Our method for ranking function synthesis can be stated as follows.
We have a finite pool of linear ranking templates.
This pool will include the multiphase, piecewise, and lexicographic ranking templates in various sizes and possibly combinations thereof.
Given a linear loop program whose termination we want to prove, we select a linear ranking template from the pool.
With this template we build the constraint \eqref{eq-rt} to the linear ranking template's parameters.
If this constraint is satisfiable, this gives rise to a ranking function according to \autoref{lem-ranking}.
Otherwise, we try again using the next linear ranking template from the pool until the pool has been exhausted.
In this case, the given linear loop program does not have a ranking function of the form specified by
any of the pool's linear ranking templates and the proof of the program's termination failed.
See \autoref{fig-synthesis-method} for a specification of our method in pseudocode.

Following related approaches~\cite{ADFG10,BMS05linrank,BMS05polyrank,CSS03,HHLP13,PR04,Rybalchenko10,SSM04},
we transform the $\exists\forall$-con\-straint \eqref{eq-rt} into an $\exists$-constraint.
This transformation makes the constraint more easily solvable
because it reduces the number of non-linear operations in the constraint:
every application of an affine-linear function symbol $f$ corresponds to a non-linear term $\tr{s_f} x + t_f$.


\subsection{Constraint Transformation using Motzkin's Theorem}
\label{ssec-constraint}

Fix a linear loop program $\loopt$ and a linear ranking template $\T$ with parameters $D$ and affine-linear function symbols $F$.
We write $\loopt$ in disjunctive normal form and $\T$ in conjunctive normal form:
\begin{align*}
\loopt(x, x') &\equiv \bigvee_{i \in I} A_i \abovebelow{x}{x'} \leq b_i \\
\T(x, x') &\equiv \bigwedge_{j \in J} \bigvee_{\ell \in L_j} \T_{j,\ell}(x, x')
\end{align*}
We prove the termination of $\loopt$ by solving the constraint \eqref{eq-rt}.
This constraint is implicitly existentially quantified over the parameters $D$ and the parameters corresponding to the affine-linear function symbols $F$.
\begin{align}
\forall x, x'.\;
\Big(
\big( \bigvee_{i \in I} A_i \abovebelow{x}{x'} \leq b_i \big)
\rightarrow
\big( \bigwedge_{j \in J} \bigvee_{\ell \in L_j} \T_{j,\ell}(x, x') \big)
\Big)
\label{eq-constraint}
\end{align}
First, we transform the constraint \eqref{eq-constraint} into an equivalent constraint of the form required by Motzkin's Theorem.
\begin{align}
\bigwedge_{i \in I} \bigwedge_{j \in J}
\forall x, x'.\;
\neg \Big(
A_i \abovebelow{x}{x'} \leq b_i
\;\land\;
\big( \bigwedge_{\ell \in L_j} \neg\T_{j,\ell}(x, x') \big)
\Big)
\label{eq-constraint2}
\end{align}
Now, \hyperref[thm-motzkin]{Motzkin's Transposition Theorem} will
transform the constraint \eqref{eq-constraint2} into an equivalent existentially quantified constraint.

This $\exists$-constraint is then checked for satisfiability.
If an assignment is found, it gives rise to a ranking function.
Conversely, if no assignment exists, then there cannot be an instantiation of the linear ranking template and thus
no ranking function of the kind formalized by the linear ranking template.
In this sense our method is sound and complete.

\begin{figure}[t]
\begin{description}
\item[Input:] linear loop program $\loopt$ and a list of linear ranking templates $\mathcal{T}$
\item[Output:] a ranking function for $\loopt$ or \texttt{null} if none is found
\end{description}
\begin{center}
\vspace{-4mm}
\begin{minipage}{82mm}
\begin{lstlisting}
foreach $\T \in \mathcal{T}$ do:
    let $\varphi$ = $\forall x, x'.\; \big( \loopt(x, x') \rightarrow \T(x, x') \big)$
    let $\psi$ = transformWithMotzkin($\varphi$)
    if SMTsolver.checkSAT($\psi$):
        let ($D$, $F$) = $\T$.getParameters()
        let $\nu$ = getAssignment($\psi$, $D$, $F$)
        return $\T$.extractRankingFunction($\nu$)
return null
\end{lstlisting}
\end{minipage}
\vspace{-7mm}
\end{center}
\caption{
Our ranking function synthesis algorithm described in pseudocode.
The function \texttt{transformWithMotzkin} transforms the $\exists\forall$-constraint $\varphi$ into an $\exists$-constraint $\psi$ as described in \autoref{ssec-constraint}.
}\label{fig-synthesis-method}
\end{figure}

\begin{theorem}[Soundness]\label{thm-soundness}
If the transformed $\exists$-constraint is satisfiable,
then the linear loop program terminates.
\end{theorem}

\begin{theorem}[Completeness]\label{thm-completeness}
If the $\exists\forall$-constraint \eqref{eq-rt} is satisfiable,
then so is the transformed $\exists$-constraint.
\end{theorem}


\section{Related Work}
\label{sec-related-work}

The first complete method of ranking function synthesis for linear loop programs through constraint solving was due Podelski and Ry\-bal\-chen\-ko~\cite{PR04}.
Their approach considers termination arguments in form of affine-linear ranking functions and requires only linear constraint solving.
We explained the relation to their method in \autoref{ex-rt-pr}.

Bradley, Manna, and Sipma propose a related approach for linear lasso
programs~\cite{BMS05linrank}.
They introduce affine-linear inductive supporting invariants to handle the stem.
Their termination argument is a lexicographic ranking function with each component corresponding to one loop disjunct.
This not only requires non-linear constraint solving, but also an ordering on the loop disjuncts.
The authors extend this approach in \cite{BMS05polyrank} by the use of \emph{template trees}.
These trees allow each lexicographic component to have a ranking function that decreases not necessarily in every step, but \emph{eventually}.

In \cite{HHLP13} the method of Podelski and Rybalchenko is extended.
Utilizing supporting invariants analogously to Bradley et al.,
affine-linear ranking functions are synthesized.
Due to the restriction to non-decreasing invariants,
the generated constraints are linear.

A collection of example-based explanations of constraint-based
verification techniques can be found in \cite{Rybalchenko10}.  This
includes the generation of ranking functions, interpolants,
invariants, resource bounds and recurrence sets.

In \cite{BAG13} Ben-Amram and Genaim discuss the synthesis of affine-linear and lexicographic
ranking functions for linear loop programs over the integers.
They prove that this problem is generally co-NP-complete and show that
several special cases admit a polynomial time complexity.


\nocite{Leike13}

\bibliographystyle{abbrv}
\bibliography{references}

\end{document}